\documentclass[runningheads]{llncs}
\usepackage{graphicx}
\usepackage{graphics}
\usepackage{array}
\usepackage{amsmath}
\usepackage{amssymb}
\usepackage{comment}
\usepackage{setspace}
\usepackage{algorithm}
\usepackage{epstopdf}       
\usepackage{enumerate}     

\usepackage{fullpage}
\usepackage{algcompatible}   
 
\date{}
\title{Complexity of Steiner Tree in Split Graphs - Dichotomy Results}
\author{Madhu Illuri, P.Renjith, N.Sadagopan} 
\institute{Indian Institute of Information Technology, Design and Manufacturing, Kancheepuram, India. \\
\email{\{coe11b012,coe14d002,sadagopan\}@iiitdm.ac.in}}
\begin{document}
\maketitle
\begin{abstract}
Given a connected graph $G$ and a terminal set $R \subseteq V(G)$, {\em Steiner tree} asks for a tree that includes all of $R$ with at most $r$ edges for some integer $r \geq 0$.  It is known from [ND12,Garey et. al \cite{steinernpc}] that Steiner tree is NP-complete in general graphs.  {\em Split graph} is a graph which can be partitioned into a clique and an independent set.  K. White et. al \cite{white} has established that Steiner tree in split graphs is NP-complete.  In this paper, we present an interesting dichotomy: we show that Steiner tree on $K_{1,4}$-free split graphs is polynomial-time solvable, whereas, Steiner tree on $K_{1,5}$-free split graphs is NP-complete.  We investigate $K_{1,4}$-free and $K_{1,3}$-free (also known as claw-free) split graphs from a structural perspective.  Further, using our structural study, we present polynomial-time algorithms for Steiner tree in $K_{1,4}$-free and $K_{1,3}$-free split graphs.  Although, polynomial-time solvability of $K_{1,3}$-free split graphs is implied from $K_{1,4}$-free split graphs, we wish to highlight our structural observations on $K_{1,3}$-free split graphs which may be used in other combinatorial problems.
\end{abstract}
\section{Introduction}
Steiner tree is a classical combinatorial optimization problem which continues to attract researchers from  both mathematics and computing.  Interestingly, this problem finds applications in Network Design, Circuit Layout Design, etc., \cite{steinerapp}. Given a connected graph $G$ and a subset of vertices (terminal set) $R \subseteq V(G)$, Steiner tree asks for a tree spanning the terminal set.  The objective is to minimize either the number of edges in the Steiner tree or the number of additional vertices ($Q \subseteq V(G) \setminus R$, also known as Steiner vertices).  It is apparent from the definition that Steiner tree generalizes well-known Minimum Spanning Tree (MST) and Shortest Path problems in general graphs \cite{cormen}.    \\
On the complexity front, Steiner tree in general graphs is NP-complete as there is a polynomial-time reduction from Exact 3 Cover \cite{garey}.  Under the assumption, NP-complete problems are unlikely to have polynomial-time algorithms, it is natural to identify the gap between polynomial-time solvability and NP-completeness by restricting the input instances.  Towards this end, many special graph classes such as chordal, bipartite, planar, split, etc., were discovered in the literature \cite{golumbic}.  Classical problems such as Vertex cover, Clique, Odd-cycle transversal have polynomial-time algorithms when the input is restricted to chordal graphs which are otherwise NP-complete for arbitrary graphs \cite{garey}.  However, other famous problems such as Hamiltonian Path (Cycle), Steiner tree, etc., remain NP-complete even on chordal graphs \cite{white,bertossi}. In fact, Steiner tree is NP-complete on Split graphs which are a strict subclass of chordal graphs \cite{golumbic}.    Steiner tree is considered to be a difficult combinatorial problem compared to other problems as it is NP-complete on almost all special graph classes.  For example, it is NP-complete on planar \cite{recti}, chordal \cite{white}, bipartite \cite{garey}, chordal bipartite \cite{mullersteinernpc} graphs.  Due to its inherent difficulty, this problem has been an active research problem in the literature for the past three decades. \\
When a combinatorial problem is NP-complete on special graph classes such as chordal and split, it is natural to restrict the input further by means of forbidden subgraphs.  For example, Hamiltonian cycle problem is NP-complete in chordal graphs, whereas it is polynomial-time solvable on interval graphs which are chordal and asteroidal-triple free \cite{Keil,hungint,panda,ibarra}.  In this paper, we revisit Steiner tree restricted to split graphs.  It is known from \cite{white}, that Steiner tree on split graphs is NP-complete.  We investigate the complexity of Steiner tree on subclasses of split graphs and present an interesting dichotomy.  Towards this end, we study  $K_{1,3}$-free (claw free) and $K_{1,4}$-free split graphs from both structural and algorithmic perspectives.   In particular, we establish the following results;
\begin{itemize}
\item Steiner tree on $K_{1,5}$-free split graphs is NP-complete.
\item Steiner tree on $K_{1,4}$-free split graphs is polynomial-time solvable.
\end{itemize}
Towards this end, we present a tight lower bound on the size of the Steiner set and our algorithm correctly produces such a Steiner set.  The above results rightly identify the gap between NP-completeness and polynomial-time solvable input instances of Steiner tree problem restricted to split graphs.  Since our contribution evolved from $K_{1,3}$-free split graphs, we highlight structural results of both $K_{1,3}$-free and $K_{1,4}$-free split graphs.  Although, the complexity of Steiner tree in $K_{1,3}$-free split graphs is inferred from $K_{1,4}$-free split graphs, out of combinatorial curiosity, we investigate both graphs from structural perspective and present polynomial-time algorithms for Steiner tree.
To the best of our knowledge, this line of investigation has not been reported in the literature.  The polynomial-time results known in the literature for Steiner tree are for trees and 2-trees \cite{steiner2tree}. \\
As far as parameterized-complexity results are concerned, in \cite{dreyfus} it is shown that Steiner tree in general is  Fixed-parameter Tractable(FPT) if the parameter is the size of the terminal set and it is $W[2]$-hard if the parameter is the size of the Steiner set \cite{saketsteiner}.   From the domain of approximation algorithms, Steiner tree has a polynomial-time approximation algorithm with ratio $2-\frac{1}{|R|}$ \cite{Garg}.  Variants of Steiner tree include Euclidean Steiner tree \cite{euclsteiner}, Rectilinear Steiner tree \cite{recti}, and Directed Steiner tree \cite{fptdirsteiner,dirsteiner}. \\
{\bf Roadmap:}  We present the structural characteristics of $K_{1,3}$-free split graphs in Section 2.  Using the structural observations made, we also present a polynomial-time algorithm to output a Steiner tree in $K_{1,3}$-free split graphs.  Structural characteristics of $K_{1,4}$-free split graph and a polynomial-time algorithm to output a Steiner tree in $K_{1,4}$-free split graphs is presented in Section 3.  Hardness result is addressed in Section 4.\\
{\bf Graph-theoretic Preliminaries:} \\
In this paper, we work with connected, simple, unweighted graphs.  Notations are as per \cite{golumbic,west}.  For a graph $G$ the vertex set is $V(G)$ and the edge set is $E(G)=\{\{u,v\}~|~u,v \in V(G)$ and $u$ is adjacent to $v$ in $G$ and $u\neq v\}$.  The neighborhood of vertex $v$ is $N_{G}(v)  = \{u~|~\{u,v\} \in E(G)\}$.  The degree of a vertex $v$ is $d_{G}(v) = |N_{G}(v)|$.   $\delta(G)=$ min $\{d_{G}(v)~|~v\in V(G)\}$.    For a graph $G$ and $S\subseteq V(G)$, $G[S]$ represents the subgraph of $G$ induced on the vertex set $S$.  The subgraph relation is represented as $G[S]\sqsubseteq G$.
A \emph{Split graph} $G=I+C$ is such that $G$ can be partitioned into an Independent Set $I$ and a Clique $C$, $V(G)=I\cup C$.  A clique $C$ is maximal if there does not exist a clique $C^{'}$ such that $C\subseteq V(C^{'})$.  For all split graphs mentioned in this paper we consider $C$ to be a maximal clique unless otherwise stated.  $K_{1,r}$ is a split graph on $r+1$ vertices such that $|C|=1$ and $|I|=r$, $E(K_{1,r})=\{\{x,v\}~|~x\in C, v\in I\}$.   \emph{$K_{1,3}$} is also termed as \emph{claw}.  \emph{Centre vertex} of a $K_{1,r}$ is the vertex of degree $r$.  A graph $G$ is $K_{1,r}$-free if $G$ forbids $K_{1,r}$ as an induced subgraph.  For a vertex $u\in C$, $N_{G}^{I}(u) = N_{G}(u)\cap I$ and $d_{G}^{I}(u) = |N_{G}^{I}(u)|$.  For $S\subseteq C$, $N_{G}^{I}(S)=\bigcup_{v\in S}N_{G}^{I}(v)$, and $d_{G}^I(S)=|N_{G}^{I}(S)|$.  For a split graph $G$, $\Delta_{G}^{I}=$ maximum$\{d_{G}^{I}(v)\},v\in C$ and $V_3=\{u\in C~|~d_{G}^{I}(u)=3\}$. 
Two edges $e_1$ and $e_2$ are non adjacent if they do not share an end vertex in common.  A set of edges $M\subseteq E(G)$ forms a matching of $G$ if every pair of edges in $M$ are non adjacent.  Maximum matching is a matching of maximum cardinality in $G$.  $\alpha(G)$ denotes the size of the maximum matching in $G$.
\vspace{-13pt}
\section{$K_{1,3}$-free Split Graphs: Structural Results}
\vspace{-10pt}
In this section, we analyze the structure of $K_{1,3}$-free split graphs and we present some interesting structural results.  Further, we show that for a claw-free split graph $G$, if $\Delta_G^I=2$, then $|I|\leq 3$.  This acts as a good handle in yeilding a linear-time algorithm for Steiner tree problem which we see in the later half of this section.  
\vspace{-5pt}
\begin{theorem} \label{thmclaw1} 	
Let $G$ be a connected split graph.  $G$ is claw free if and only if one of the following conditions hold. \\
1. $\Delta^{I}_{G}\leq 1$    \\
2. $\Delta^{I}_{G}=2$ and  for every $u,v\in C$ such that $d^{I}_{G}(u)=2$, $N^{I}_{G}(u)\cap N^{I}_{G}(v)\neq \emptyset$
\end{theorem}	
\begin{proof}
\emph{Necessity:} Suppose $\Delta_G^I \geq 3$, and let $v\in C$ has at least 3 neighbours, say $x,y,z\in I$. Then the set $\{v,x,y,z\}$ forms a claw in $G$ with $v$ as its centre vertex.  It follows that if $G$ is claw-free, then $\Delta_G^I \leq 2$.  Now suppose $\Delta_G^I = 2$.  Let $u\in C$ such that  there exist vertices $x, y\in I, \{x, y\}\subseteq N_G^I(u)$.  We assume on the contrary that there exist $v\in C, v\neq u$ such that $N_G^I(u) \cap N_G^I(v) =\emptyset$.  Since $C$ is a clique, $\{u,v\}\in E(G)$. It follows that vertices  $\{u,x,y,v\}$ forms a claw in $G$ with $u$ as its centre, a contradiction. This proves Condition \emph{2}, and completes the proof of the forward direction. \\
\emph{Sufficiency:} On the contrary assume that $G$ is not claw free.  No claw in $G$ can have its centre vertex in the set $I$, since for any $v$ in $I$, the set $N_G(v)\subseteq C$ and hence induces a clique in $G$.  So every claw in $G$ has its centre vertex in the set $C$.
Consider a claw with the vertex set $\{v,x,y,z\}$, with the centre $v$ being in $C$.  No two of the other three vertices of the claw can be in $C$, because then there would be an edge between them. So at most one of $\{x,y,z\}$ is in $C$, and the rest (of which there are at least two) are in $I$. It follows that if $G$ contains a claw, then $\Delta_G^I \geq 2$. Equivalently, if $\Delta_G^I \leq 1$ then $G$ is claw-free.
Finally, consider the case where $\Delta_G^I = 2$.  Suppose the vertex set $\{v,x,y,z\}$ induces a claw in $G$, with its centre vertex being $v$. Then $v$ is in $C$, and at least two of $\{x,y,z\}$ are in $I$, as we argued above.  Since $\Delta_G^I = 2$ we get that exactly  two of $\{x,y,z\}$, say $x$ and $y$, are in $I$. Then $z$ is in $C$, and $\{x,z\},\{y,z\}\notin E(G)$.   It follows that $N^{I}_{G}(v)\cap N^{I}_{G}(z)= \emptyset$ which is a contradiction to Condition \emph{2}.  Therefore, our assumption that there exist a claw in $G$ is wrong, and this completes the sufficiency.  Therefore, the theorem follows.$\hfill \qed$
\end{proof}
\vspace{-15pt}
\begin{lemma} \label{lemvileq3}
For a claw-free split graph $G$,  if $\Delta_G^I=2$, then $|I|\leq 3$. 
\end{lemma} \vspace*{-10pt}
\begin{proof} 
Since $\Delta_G^I=2$, let there exist a vertex $v \in C$ such that $d_{G}^{I}(v) = 2$.  On the contrary, assume that $|I|> 3$, that is, $\{a,b,c,d\}\subseteq I$ such that $N_{G}^{I}(v) = \{a,b\}$.  Let $X = N_{G}(a)$ and $Y = N_{G}(b)$ as shown in Figure \ref{thmclaw1fig}.  If there exist a vertex $t\in C$ such that $t\notin X$, and  $t\notin Y$, then vertices $\{a,b,v,t\}$ induces a claw.  Therefore, $C = X\cup Y$.  If $X\subseteq Y$, then $C\cup \{b\}$ induces a larger clique, which is a contradiction to the assumption on the maximality of clique $C$.  Therefore, $X\not\subseteq Y$ and similarly, $Y\not\subseteq X$.  It follows that, $X-Y\neq\emptyset$ and $Y-X\neq\emptyset$.  For every vertex $v\in X\cap Y$, $\{v,c\}\notin E(G)$ and $\{v,d\}\notin E(G)$ otherwise, $N_{G}^{I}(v)\cup \{v\}$ induces $K_{1,3}$.  
\vspace{-10pt}
\begin{figure}%
\begin{center}
\includegraphics[scale=1.2]{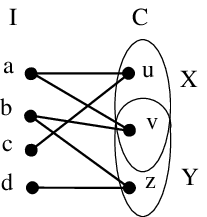} \hspace{50pt}
\includegraphics[scale=1.2]{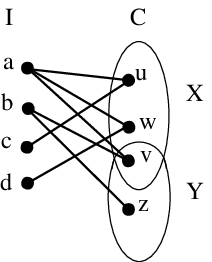}\\ Case \emph{1} \hspace{100 pt} Case \emph{2}
\end{center}
\vspace{-20pt} \caption{An illustration for the proof of Lemma \ref{lemvileq3}}%
\label{thmclaw1fig}%
\end{figure}
\vspace{-5pt}
Therefore, the vertices $c,d$ can have adjacency in two disjoint sets $X-Y$ or $Y-X$.  \\
\emph{Case 1:} $N_{G}(c)\cap X\neq \emptyset$ and $N_{G}(d)\cap Y\neq \emptyset$.   Edge $\{u,c\}\in E(G)$ where $u\in X-Y$ and $\{z,d\}\in E(G)$ where $z\in Y-X$.  Observe that $\{a,z\},\{c,z\}\notin E(G)$ otherwise $N_{G}^{I}(z)\cup \{z\}$ induces $K_{1,3}$.  Similarly, $\{d,u\}\notin E(G)$ otherwise $N_{G}^{I}(u)\cup\{u\}$ induces $K_{1,3}$.  From the discussion, it follows that the vertices $\{u,a,c,z\}$ induces a claw, which is a contradiction.  Similar argument holds for $N_{G}(c)\cap Y\neq \emptyset$ and $N_{G}(d)\cap X\neq \emptyset$.
\\
\emph{Case 2:} $N_{G}(c)\subseteq X$ and $N_{G}(d)\subseteq X$.  Let $\{c,u\}, \{d,w\}\in E(G)$ such that $u,w\in X-Y$.  Note that there exist at least one vertex $z\in Y-X$.  If $\{c,z\}\notin E(G)$, then the vertices $\{u,a,c,z\}$ induces a claw.  If $\{c,z\}\in E(G)$, then the vertices $\{w,a,d,z\}$ induces a claw as, $\{d,z\}\notin E(G)$.   The argument is symmetric for $N_{G}(c)\subseteq Y$ and $N_{G}(d)\subseteq Y$.  \\
Cases \emph{1} and \emph{2} give a contradiction to the fact that $G$ is claw free.  Therefore, our assumption that $|I|> 3$ is wrong, and hence, the lemma follows. \hfill \qed
\end{proof}
\vspace{-20pt}
\subsection{Application: Steiner tree in $K_{1,3}$-free Split Graphs} \label{section2.1}
Using the structural results presented in Section \emph{2}, in this section, we present a polynomial-time algorithm to find minimum Steiner tree in $K_{1,3}$-free split graphs.   Optimum version of Steiner tree problem is defined as follows;
\vspace{-8pt}
\begin{center}
\fbox{\parbox[c][][c]{0.8\textwidth}{    
\emph{OPT Steiner tree(G,R)}\\
Instance:	Graph $G(V,E)$, Terminal Set $R\subseteq V(G)$\\
Question:	Find a minimum cardinality set $S\subseteq V(G)\backslash R$ such that $G[S\cup R]$ is connected?} }	
\end{center}
We here consider the Steiner tree problem on split graph $G^0=I^0+C^0$.  Due to pruning, we iteratively construct split graphs $G^1$, $G^2$ from the input graph $G=G^0$.   We simplify	the input by pruning the vertices which are not part of any optimum solution.  The pruned graph $G^1$ is the graph induced on the vertex set $V(G^0)\backslash (S_1\cup S_2\cup S_3)$.  Clearly, $G^{1}\sqsubseteq G^0$ and let $G^{1}=I^{1}+C^{1}$.   We prune three sets of vertices $S_1,S_2,S_3$ one after the other and are defined as follows.
$S_1=$ $\{a\in I^0$ $|~a\notin R\}$.
$S_2=$ $\{u\in C^0$ $|~u\notin R$ and $N_{G}^{I^0}(u)\cap R=\emptyset\}$.
Let $R^{'}=\{v\in C^0~|~v\in R\}$.  $S_3=$ $\bigcup\limits_{v\in R^{'}}\{ v \} \cup N_{G}^{I^0}(v)$. 
Consider the Steiner tree optimization problems $P_1$, and  $P_2$ defined as follows.\\
$P_1$: OPT Steiner tree($G^0,R$) \\
$P_2$: OPT Steiner tree($G^{1},R\backslash S_3$) \\
\vspace{-15pt}
\begin{lemma} \label{lem1}
An optimum solution $Q$ to $P_2$ is also an optimum solution to $P_1$.
\end{lemma}
\begin{proof}
Note that the first two sets $S_1,S_2$ pruned from $G^0$ are not part of any optimum solution.   $S_3\subseteq R$ induces a connected subgraph of $G^0$ which is also pruned to obtain $G^{1}$.    If $V(G^{1})\cap R=\emptyset$, then Steiner set of $P_2$ is empty.  i.e., $R$ induces a connected subgraph of $G^0$.   On the other hand if  $V(G^{1})\cap R\neq\emptyset$, then there exist at least one vertex $v\in C^1$ in the Steiner set  $Q$ of $P_2$.   $Q\subseteq C^1$ connects all terminal vertices $R\backslash S_3$.   If $S_3\neq\emptyset$, then there exist at least one vertex $u\in S_3$ such that $u\in C^0$ and $u\in R$.  $\{u,v\}\in E(G^0)$ and therefore,  $Q\cup R$ induces a connected subgraph of $G^0$ and $Q$ is a minimum Steiner set for $P_1$.  Hence, the lemma follows. \hfill \qed
\end{proof}
\vspace{-22pt}
\subsubsection{A polynomial-time algorithm to find a minimum Steiner tree }
Given a $K_{1,3}$-free split graph $G^0$ with terminal vertex set $R\subseteq V(G^0)$,  we present a polynomial-time algorithm to find a minimum Steiner tree.  
As part of preprocessing step, we prune the sets $S_1,$ and $S_2$, which are not part of any optimum solution.  Further, we delete terminals which are in $C$, and their neighbours in $I$, namely the set $S_3$.  Now we have an instance of Steiner Tree in claw-free split graphs where all the terminals are in the independent set.  An optimum solution to the pruned graph is also an optimum solution to the original graph by the previous lemma.  We now present a sketch of algorithm and the detailed one is presented in Algorithm \ref{alg1}.  If $\Delta_G^I = 0$, then the instance is trivial. If $\Delta_G^I = 1$, then Steiner set should contain one neighbor vertex in $C$ of each terminal in $I$. 
In the remaining case, $\Delta_G^I = 2$ and therefore, by Lemma 1 $|I| \leq 3$. The only non-trivial case is when $|I| = 3$. From the constraints of the instance, we know that it is necessary and sufficient to pick exactly two Steiner vertices from $C$ in this case.
\begin{algorithm}[!h] 
\caption{Steiner tree in Claw free Split graphs.  Steiner\_tree($G^0,R$)} \label{alg1}
\begin{algorithmic}[1]
\STATEx {/*$G^0$ \texttt{is a claw-free split graph and} $R$ $\subseteq$ $V(G^0)$ \texttt{is the set of terminal vertices} */}
\STATE { Find the pruned graph $G^{1}$=Pruning($G^0,R$) }
\STATE {Initialize the output set of Steiner vertices $S=\emptyset$ and unmark every vertices in $I^1\subseteq V(G^{1})$}
\IF {  $\Delta_{G^{1}}^{I}=1$}
	\FOR {every unmarked vertex $d\in I^1$}
		\STATE {include $w\in C^1$ in $S$ where $\{d,w\}\in E(G^{1})$.}
		\STATE {mark vertex $d$.}
  \ENDFOR
\ELSE
 \STATE {include vertex $x\in C^1$ in $S$ where $|N_{G^{1}}^{I}(x)|=2$. i.e.,$N_{G^{1}}^{I}(x)=\{a,b\}$}
 \IF {$|I^1|=3$. i.e.,$I^1=\{a,b,c\}$}
   \STATE {include $y\in C^1$ in $S$ where $\{c,y\}\in E(G^{1})$}
 \ENDIF
\ENDIF
\STATE {Run standard Breadth First Search in the graph $G[S\cup R]$ and output the BFS tree.}
\end{algorithmic}
 \end{algorithm}
\vspace{-18pt}
 \begin{algorithm}[!h]
\caption{Pruning the input instance of Steiner tree.  Pruning($G^0,R$) } \label{alg2}
\begin{algorithmic}[1]
\STATEx {/* $G^0:$\texttt{input claw-free split graph,} $R:$\texttt{set of terminal vertices} */}
\STATE{Find the sets $S_1,S_2,S_3$ in order and prune those vertices from $G^0$. i.e., $G^{1}=G^0\backslash S$ where $S=S_1\cup S_2\cup S_3$}
\STATE {Return the pruned graph $G^{1}$.}
\end{algorithmic}
\end{algorithm}
\vspace{-25pt}
\subsubsection{Proof of Correctness of Algorithm \ref{alg1}} \noindent \\
By Lemma \ref{lem1}, a minimum Steiner set of pruned graph $G^{1}$  is an optimum Steiner set for $G^0$.   Therefore, pruning in Step $1$ is a solution preserving operation.  We present a case analysis to show that our algorithm outputs a minimum Steiner tree of a claw-free split graph.  
\\
\textit{Case 1:} $\Delta^{I}_{G^{1}}\leq1$.  Note that for every vertex $d\in I^1$, Step $5$ includes exactly one vertex $w\in N_{G^{1}}(d)$ in $S$, which is a minimum Steiner set.
\\
\textit{Case 2:} $\Delta^{I}_{G^{1}}=2$. 
Observe $|I|\leq 3$ by Lemma \ref{lemvileq3}.  $|S|=1,2$ for $|I|=2,3$, respectively, which is done by Steps $9,11$.  Therefore, $S$ is a minimum Steiner set for $G^{1}$, and by Lemma \ref{lem1}, $S$ is also a minimum Steiner set for $G^0$.  Step $14$ outputs a Steiner tree by running standard Breadth First Search algorithm on $G[S\cup R]$.
\vspace{-18pt}
\subsubsection{Run Time Analysis} \noindent \\
We represent the input claw-free split graph using an adjacency list, as we can easily find a neighbor of a given vertex.  Vertices in adjacency list are arranged such that $C^0$ follow $I^0$.  Intuition behind this ordering is that, first neighbor of a vertex $v\in C^0$ encountered in the list is always a vertex $u\in I^0$, if it exists.  If $\Delta^{I}_{G^{1}}=1$, then $u\in N_{G^{1}}(v)$ can be determined in constant time.  Therefore, Algorithm \ref{alg1} takes linear time $O(n), n=|V(G^0)|$ to output a minimum Steiner set.
\vspace{-15pt}
\section{$K_{1,4}$-free Split Graphs: Structural Results}
\vspace{-10pt}
In this section, we first analyze the structure of $K_{1,4}$-free split graphs.  Subsequently we investigate Steiner tree problem restricted to $K_{1,4}$-free split graphs.  Towards this end, we give a nice bound on the cardinality of any minimum Steiner set.  Further, we present a structural characterization of $K_{1,4}$-free split graph meeting the bound.  Interestingly, the characterization yields a polynomial-time algorithm to output a minimum Steiner tree, which we shall present in Section \ref{K14algsection}.\\
Before we present the structural results, we introduce some additional terminologies.   
A split graph $G$ is a $l$-split graph if $\Delta_{G}^{I}=l$.  Note that a $K_{1,4}$-free split graph is a $l$-split graph for some $l, 0\leq l\leq 3$, and the converse does not always hold.  In a split graph $G$, closed neighborhood of a vertex $u\in C$ is $[N(u)]=\{u\}\cup N_G^I(u)$.   
For a $l$-split graph $G=I+C, 0\leq l\leq 2$, we construct a \emph{labeled} graph $M$ such that $V(M)=I$ and $E(M)=\{\{a,b\}~|~a,b\in I$ and $N_{G}(a)\cap N_{G}(b)\neq \emptyset\}$ and label the edge $\{a,b\}$ as $v_{ab}$.  Note that $v$ in $v_{ab}$ denotes a vertex $v\in N_{G}(a)\cap N_{G}(b)$.  Also, we pick exactly one $v\in N_{G}(a)\cap N_{G}(b)$ to label the edge $\{a,b\}$.  For any edge set $E^*\subseteq E(M)$, we define the \emph{corresponding vertex} set $V^*$ as follows. 
Corresponding to each edge $\{a,b\}\in E^*$, include exactly one vertex $v\in N_{G}(a)\cap N_{G}(b)$ in $V^*$.  It follows that, $V^*\subseteq C$ and $|V^*|=|E^*|$.  Clearly, $|V^*|\leq |E^*|$ as we are including not more than one vertex in $V^*$ corresponding to each edge in $E^*$.  Suppose $|V^*|<|E^*|$, then there exist at least two edges labelled $v_{ab}, v_{cd}$ in $E^*$ such that $v\in N_{G}(a)\cap N_{G}(b)$ and $v\in N_{G}(c)\cap N_{G}(d)$.  Since edges $\{a,b\}, \{c,d\} \in E^*$ can share atmost one vertex in common, it follows that, $d^I_G(v)\geq 3$, which is a contradiction as $G$ is $l$-split, $l\leq2$ and $M$ is the labelled graph of $G$.  Therefore, $|V^*|=|E^*|$.  We also define the \emph{Corresponding clique} set $V^c$ of a vertex set $V'\subseteq I$ as follows.  Corresponding to each vertex $u\in V'$, include exactly one vertex $w$ in $V^c$ such that  $\{u,w\}\in E(G)$.  Clearly, $V^c\subseteq C$ and $|V^{c}|\leq|V'|$.  For a $1$-split graph, $|V^c|=|V'|$.
%
We now present some structural observations on  $K_{1,4}$-free split graphs. 
\vspace{-5pt}
\begin{lemma} \label{lem3}
Let $G$ be a $3$-split graph.  $G$ is $K_{1,4}$ free if and only if for every $u\in V_3$ and for every $v\neq u\in C$, $N_{G}^{I}(u)\cap N_{G}^{I}(v)\neq\emptyset$. 
\end{lemma}
\begin{proof}\emph{Necessity: } On the contrary, let us assume there exist $v \in C$ such that $N_{G}^{I}(u)\cap N_{G}^{I}(v)=\emptyset$.  Since $d_{G}^{I}(u) =3$, vertices $\{u,v\}\cup N_{G}^{I}(u)$ induces a $K_{1,4}$, which is a contradiction and the necessary condition follows.\\
%
\emph{Sufficiency: } On the contrary, assume that $G$ is not $K_{1,4}$ free and there exists a $K_{1,4}$ induced on $\{u,v,w,x,y\}$ with $u$ as the centre vertex.   No $K_{1,4}$ in $G$ can have its centre vertex in the set $I$, since for any $u$ in $I$, the set $N_{G}(u)$ is a subset of the set $C$ and hence induces a clique in $G$.  So every $K_{1,4}$ in $G$ has its centre vertex in the set $C$ particularly, $u\in C$.  Since $G$ is a $3$-split graph, $d_{G}^{I}(u)= 3$.  
This implies that there exist at least one vertex of $K_{1,4}$, say $v\in C$, and $u\in V_3$.  It follows that $N_{G}^{I}(u)\cap N_{G}^{I}(v)=\emptyset$, which is a contradiction and the sufficiency follows.  This completes the proof of the lemma.  \hfill \qed
\end{proof}
\vspace{-15pt}
\begin{corollary} \label{cor1}
Let $G$ be a $K_{1,4}$-free  $3$-split graph.  For any $v\in C$, the graph $H$  induced on the vertex set $V(G)\setminus N_{G}^I(v)$ is a $l$-split graph for some $0\leq l\leq 2$.
\end{corollary}
On the contrary, suppose there exists a vertex $w\in C$ such that $d_{H}^I(w)=3$.  i.e., $w\in V_3$.  It follows that $N_{G}^{I}(w)\cap N_{G}^{I}(v)=\emptyset$.  By previous lemma,  $N_{G}^I(w)\cup \{w,v\}$ induces a $K_{1,4}$, which is a contradiction. \hfill \qed 
%
\noindent
\begin{corollary}  \label{cor2}
Let $G$ be a $K_{1,4}$-free split graph and $v\in C$. If $N_{G}^{I}(v)=\{v_{1},v_{2},v_{3}\}$, then $N_{G}(v_{1})\cup N_{G}(v_{2})\cup N_{G}(v_{3})=C$.
\end{corollary}
\begin{proof}  By Lemma \ref{lem3}, for every  $u\in C$, $N_{G}^{I}(v)\cap N_{G}^{I}(u) \neq \emptyset$.  This implies that for every  $u\in C$,  $\{v_1,v_2,v_3\}\cap N_{G}^{I}(u)\neq\emptyset$.  It follows that $N_{G}(v_1)\cup N_{G}(v_2)\cup N_{G}(v_3)=C$.  \hfill \qed   
\end{proof}
\vspace{-8pt}
Now onwards, we investigate the Steiner tree problem on $K_{1,4}$-free split graphs.  For our discussions on Steiner tree problem, we fix the terminal set $R$ to be $I$.  Observe that $l$-split graphs for $l=1,2$ are $K_{1,4}$-free split graphs.  If $G$ is a $1$-split graph, then there does not exist a vertex $v\in C$ such that $d_{G}^I(v)\geq 2$.  Therefore,  the corresponding clique set of $I$ forms the minimum Steiner set $S$ of $G$ where $|S|=|I|$.   We shall now consider $2$-split graphs for discussions.  For a $2$-split graph  $G$, recall that the labelled graph $M$ is such that $V(M)=I $,  $ E(M)=\{\{a,b\}~|~a,b\in I$ and there exist $v\in C$ such that $\{a,b\}=N_{G}^I(v)\}$.  
The following lemma gives the cardinality of a minimum Steiner set of any $2$-split graphs.
\begin{lemma}\label{2splitalpha}
Let $G$ be a $2$-split graph, and $M$ be the labeled graph of $G$ with $\alpha(M)=k$.  Then any minimum Steiner set $S$ of $G$ is such that $|S|=|I|-k$.
\end{lemma}
\begin{proof}
 If $M$ is a connected graph, then the minimum Steiner set in $G$ corresponds to the minimum edge cover in $M$.  For any graph $M$ with maximum matching $P$, the cardinality of minimum edge cover is $|V(M)|-|P|$.  Therefore, a minimum Steiner set $S$ is such that $|S|=|V(M)|-|P|=|I|-k$.  If $M$ is not connected, let $C_1,C_2,\ldots,C_r$ be the components such that $C_1,C_2,\ldots,C_i, i\leq r$ are non-trivial components with at least one edge and $C_{i+1},C_{i+2},\ldots,C_r$ are trivial ones.  For components $C_1,C_2,\ldots,C_i$, we find the maximum matching $P$ where $k=|P|$ and $Q\subseteq C$ be the corresponding vertex set of the matching $P$.  Clearly, $|N_{G}^I(Q)|=2|Q|=2|P|=2k$.   Let $Q'$ be the corresponding clique set of $I\backslash N_{G}^I(Q)$.  From the definition of the corresponding clique set, $|Q'|\leq |I\backslash N_{G}^I(Q)|$.   Note that, there does not exist two vertices $x,y\in I\backslash N_{G}^I(Q)$ such that $N_{G}(x)\cap N_{G}(y)\neq \emptyset$, otherwise it contradicts the maximality of $P$.    Since there does not exist the possibility to have two such vertices $x,y\in  I\backslash N_{G}^I(Q)$, it follows that $|Q'|= |I\backslash N_{G}^I(Q)|$ and the graph induced on $V(G)\backslash N_G^I(Q)$ is a $1$-split graph.  Therefore, $|Q'|=|I|-2k$, and $I\setminus N_{G}^I(Q)\subseteq N_{G}^I(Q')$.  It follows that the set $S=Q'\cup Q$ forms a Steiner set of $G$ and $|S|=|I|-2k+k=|I|-k$.  
 \hfill \qed
 \end{proof}
\begin{lemma}
For any $2$-split graph $G$, OPT Steiner tree problem is polynomial-time solvable.
\end{lemma}
\begin{proof}
Finding the labeled graph $M$ of $G$, incurs $O(n)$ effort where $n=|V(G)|$.  Maximum matching $P$ of $M$ can be found in $O(n^{\frac{3}{2}})$ time.   Note that the corresponding vertex set $Q$ of $P$ can be found in linear time.  Similarly, the corresponding clique set also can be obtained in linear time.  Therefore, the overall running time for finding the Steiner set is $O(n^{\frac{3}{2}})$ and OPT Steiner tree in any $2$-split graph is polynomial-time solvable. \hfill \qed
 \end{proof}
The following lemma characterizes a special $2$-split graph constructed from a $3$-split graph.  Particularly, Lemma \ref{lem2} gives an upper bound on the matching size of the labelled graph of a $2$-split graph.  
\begin{lemma} \label{lem2}
Let $G^1=I^1+C^1$ be a  $K_{1,4}$-free $3$-split graph.  For any  $x\in V_3$, let $G^2$ be the graph induced on $V(G^1)\backslash N^I_{G^1}(x)$, and $M$ be the labelled graph of $G^2$.  Then size of any maximum matching $\alpha(M)\leq 2$
\end{lemma}
\begin{proof}
 Recall from Corollary \ref{cor1} that, $G^2$ is a $l$-split graph for some $0\leq l\leq 2$.  On the contrary, let $\alpha(M)\geq 3$.  Let vertices $\{a,b,c,d,e,f\}\subseteq V(M)$ be those vertices participating in the matching of size at least $3$ such that $\{u,v,w\}\subseteq C^1$ and $\{a,b\}\subseteq N^I_{G^1}(u)$, $\{c,d\}\subseteq N^I_{G^1}(v)$, $\{e,f\}\subseteq N^I_{G^1}(w)$ as shown in Figure \ref{lemmamatfig}.  Clearly, from Lemma \ref{lem3}, $N_{G^1}^I(x)\cap N_{G^1}^I(u)\neq \emptyset$.  Similarly, $N_{G^1}^I(x)\cap N_{G^1}^I(v)\neq \emptyset$ and $N_{G^1}^I(x)\cap N_{G^1}^I(w)\neq \emptyset$.  We consider the following scenario.
\begin{figure}[!h]
\begin{center}
 \includegraphics[scale=1.1]{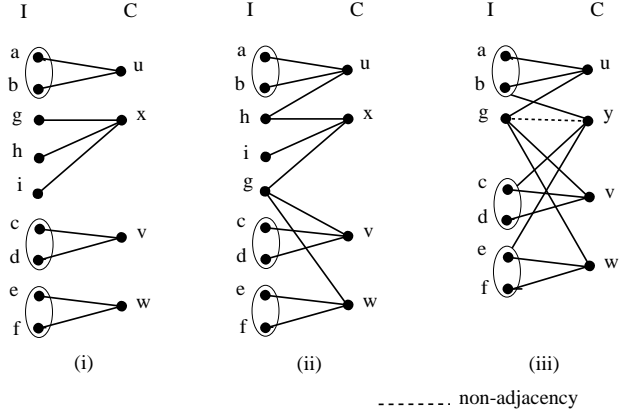}
\end{center} \vspace{-20pt}
\caption{An illustration for the proof of lemma \ref{lem2}}%
\label{lemmamatfig}%
\end{figure}
\\Suppose $\{g,v\}\in E(G^1)$ and $\{g,u\}\notin E(G^1)$.  Since $N_{G^1}^I(x)\cap N_{G^1}^I(u)\neq \emptyset$, without loss of generality, $\{h,u\}\in E(G^1)$.  Observe that $\{u,v\}\cup N_{G^1}^I(v)$ induces a $K_{1,4}$.  Therefore, $\{g,u\}\in E(G^1)$.  Similar argument holds true for $w$ and $\{g,w\}\in E(G^1)$. 
Since the clique $C^1$ is maximal, $g$ is not adjacent to all vertices of $C^1$, and therefore there exist $y\in C^1$ such that $\{g,y\}\notin E(G^1)$.  Clearly,  $N_{G^1}^{I}(y)\cap N_{G^1}^{I}(u)\neq\emptyset$, $N_{G^1}^{I}(y)\cap N_{G^1}^{I}(v)\neq\emptyset$ and  $N_{G^1}^{I}(y)\cap N_{G^1}^{I}(w)\neq\emptyset$.  Observe that in $G^2, d_{G^2}^{I}(y)=3$, and $G^2$ is not a $l$-split graph, $l\leq 2$.  This is a contradiction to Corollary \ref{cor1}.  
It follows that our assumption $\alpha(M)\geq 3$ is wrong and therefore, $\alpha(M)\leq 2$.  This completes the proof of the lemma. \hfill \qed 
\end{proof} 
We now present some structural observations pertaining to $3$-split graphs.  
\begin{lemma} \label{lem4}
For a $K_{1,4}$-free $3$-split graph $G^1$, any Steiner set $S$ of $G^1$ is such that $|S|\geq |I^1|-5$.
\end{lemma}
\begin{proof}
Observe that $V_3\neq\emptyset$ as $G^1$ is $3$-split.  For any $v\in V_3$, $G^2$ is the graph induced on $V(G^1)\backslash N^I_{G^1}(v)$.  By Corollary \ref{cor1}, $G^{2}$ is a $l$-split graph, $l\leq 2$.  Let $S^2$ be the minimum Steiner set of $G^2$ such that $N_{G^2}^I(S^2)=I^2$ and $|I^2|=|I^1|-3$.  If $M$ is the labeled graph of $G^2$, then by Lemma \ref{lem2}, $\alpha(M)\leq 2$.  Let $\{\{a,b\}, \{c,d\}\}\subseteq E(M)$ be the matching edges of a maximum matching in $M$.  Observe that there exist two vertices $v_1,v_2\in C^2$ such that $N_{G^2}^I(v_1)=\{a,b\}$, $N_{G^2}^I(v_2)=\{c,d\}$.  Notice that for each vertex $w\in I^2\backslash \{a,b,c,d\}$, there exist a vertex  $u\in C^2\backslash \{v_1,v_2\}$ in $S^2$ such that $\{w,u\}\in E(G^2)$.  The graph induced on $V(G^2)\backslash \{a,b,c,d\}$ is a $1$-split graph,  $|S^2|\geq |I^2|-4+2$ and it follows that $|S^2|\geq |I^1|-3-2=|I^1|-5$.   It can be concluded that $|S|\geq |I^1|-5$ as $|S|\geq |S^2|$.  This completes the proof of the lemma.  \hfill \qed
\end{proof}
We below characterize $K_{1,4}$-free $3$-split graphs based on the cardinality of a minimum Steiner set.  In particular, in Theorem \ref{vi-4}, we characterize  $K_{1,4}$-free split graphs whose minimum Steiner set is $|I^1|-4$, and in Theorem \ref{vi-3}, we characterize  $K_{1,4}$-free split graphs whose minimum Steiner set is $|I^1|-3$.  To present Theorem \ref{maxn-4} to Theorem \ref{bound}, we fix the following notation.  Let $G^1=I^1+C^1$ be a  $K_{1,4}$-free $3$-split graph.  For any  $u\in V_3$, let $G^2=I^2+C^2$ be the graph induced on $V(G^1)\backslash N^I_{G^1}(u)$, and $M$ be the labelled graph of $G^2$.  In Theorem \ref{maxn-4}, we present a stronger result of Lemma \ref{lem4}.   

\begin{theorem} \label{maxn-4}
 For a $K_{1,4}$-free $3$-split graph $G^1$, any minimum Steiner set $S$ of $G^1$ is such that $|S|\geq|I^1|-4$.
\end{theorem}
\begin{proof}
On the contrary assume that there exist a minimum Steiner set $S\subseteq C^1$ such that $|S|\leq |I^1|-5$.  \\
Suppose that $S\cap V_3=\emptyset$.  Note that $V_3\neq\emptyset$, say $u\in V_3$ and for every vertex $z\in S$, $N_{G^1}^{I}(z)\cap N_{G^1}^{I}(u)\neq\emptyset$ as per Lemma \ref{lem3}.  i.e., for every $z\in S$, there exist an edge $\{z,i\}\in E(G^1)$, where $i\in N_{G^1}^{I}(u)$.  The graph $G^{2}=I^2+C^{2}$  induced on $V(G^1)\backslash N_{G^1}^{I}(u)$ is a $l$-split graph, $l\leq 2$ by Corollary \ref{cor1}.   Consider the Steiner set $S^2\subseteq C^2$ of $G^2$ such that $N_{G^2}^I(S^2)=I^2$.  Note that $|S^2|=|I^2|$ as $G^2$ is a $l$-split graph, $l\leq 2$ and for each vertex $w\in S^2$, $N_{G^1}^I(w)\cap N_{G^1}^I(u)\neq \emptyset$.  Notice that $|I^2|=|I^1|-3$ and  $|S|\geq |S^2|$ implies that $|S|\geq |I^1|-3$.  This shows that $S\cap V_3=\emptyset$ is not possible.  \\
Next we shall consider the scenario $S\cap V_3\neq\emptyset$.   Consider the $l$-split graph, $l\leq 2$ $G^{2}$ induced on $V(G^1)\backslash N_{G^1}^{I}(u)$ where $u\in S\cap V_3$.  Let the labeled graph of $G^2$ be $M$.  For $S^{'}=S\backslash \{u\}$ and $V^{'}=I^1\backslash N_{G^1}^I(u)$ note that $N^{I}_{G^2}(S^{'})=V^{'}$.  Clearly, $|S^{'}|=|S|-1\leq |I^1|-5-1$ and $|V^{'}|=|I^1|-|N_{G^1}^{I}(u)| = |I^1|-3$.   We now claim that there exist at least $3$ vertices say  $\{v_1,v_2,v_3\}\subseteq S^{'}$ such that $d^{I}_{G^2}(v_i)=2$, $i=1,2,3$ and $N^{I}_{G^2}(v_i)\cap N^{I}_{G^2}(v_j)=\emptyset$, $1\leq i\neq j\leq 3$.  Suppose if there exist at most two vertices $v_1,v_2\in S^{'}$ such that $d_{G^2}^I(v_1)=d_{G^2}^I(v_2)=2$ and $N^{I}_{G^2}(v_1)\cap N^{I}_{G^2}(v_2)=\emptyset$, then observe that $|V^{'}|\leq |S^{'}|+2$.  It follows that $|V^{'}|\leq |I^1|-6+2= |I^1|-4$, which is a contradiction as $|V^{'}|$ is $|I^1|-3$.  Therefore, there exist at least $3$ vertices $v_1,v_2,v_3\in S^{'}$ such that $d^{I}_{G^2}(v_i)=2$, $i=1,2,3$ and $N^{I}_{G^2}(v_i)\cap N^{I}_{G^2}(v_j)=\emptyset$, $1\leq i\neq j\leq 3$.   Consider the labeled graph $M$ of $G^2$.  There exist $\{a,b,c,d,e,f\}\subseteq I^2$ such that $\{a,b\}=N_{G^2}^I(v_1)$, $\{c,d\}=N_{G^2}^I(v_2)$, $\{e,f\}=N_{G^2}^I(v_3)$.  It follows that $\{a,b\}$, $\{c,d\}$, $\{e,f\}$ forms a matching of size $3$ in $M$ which is a contradiction to Lemma \ref{lem2}.  Therefore our assumption is wrong and $|S|\geq |I^1|-4$.  This completes the proof. \hfill\qed
\end{proof}
\vspace{-5pt}
We show in Theorem \ref{vi-4} that the lower bound in Theorem \ref{maxn-4} is tight.   
\vspace{-2pt}
\begin{theorem} \label{vi-4}
For any minimum Steiner set $S$ of $G^1$, $|S|=|I^1|-4$ if and only if $\alpha(M)=2$.
\end{theorem}
\begin{proof}
\emph{Necessity:}  
If $S\cap V_3=\emptyset$, then for every vertex $z\in S$, $N_{G^1}^{I}(z)\cap N_{G^1}^{I}(u)\neq\emptyset$.  i.e., for every $z\in S$, there exist an edge $\{z,i\}\in E(G^1)$, where $i\in N_{G^1}^{I}(u)$.  Similar to the proof of Theorem \ref{maxn-4}, it follows that $|S|\geq |I^1|-3$.  Therefore, $S\cap V_3\neq\emptyset$.  Let $u\in S\cap V_3$ and the graph $G^2$ induced on vertex set $V(G^1)\backslash N_{G^1}^I(u)$ is a $l$-split graph, $l\leq 2$ by Corollary \ref{cor1}.  Let $S^{'}=S\backslash \{u\}$.  Clearly in $N_{G^2}^I(S^{'})=I^2$.  Note that $|S^{'}|=|I^1|-5$ and $|I^2|=|I^1|-3$.  This implies that there exist a matching of size at least $2$ in $M$.  From Lemma \ref{lem2}, $\alpha(M)\leq 2$.  Therefore, $\alpha(M)=2$. \\
\emph{Sufficiency:} Let $\{a,b\}, \{c,d\}\in E(M)$ be the edges that form the matching of size $2$ such that label($\{a,b\})=v_{ab}$ and label($\{c,d\})=w_{cd}$.  Clearly, $d_{G^1}^I(u)=d_{G^1}^I(v)=d_{G^1}^I(w)=3$ and $|N_{G^1}^I(X)|=7$, where $X=\{u,v,w\}$.   Let $Y$ be the corresponding clique set of $I^1\backslash N_{G^1}^I(X)$.  Observe that $X\cup Y$ forms a Steiner set of $G^1$, $|Y|=|I^1|-7$ and  $|X\cup Y|=|I^1|-4$.  This completes the proof.  \hfill\qed
\end{proof}
\vspace{-7pt}
Apart from the labelled graph $M$, we make use of one more labelled graph in Theorem \ref{vi-3}, which is defined as follows.  $H^2=I^2_H+C^2_H$ is the $l$-split graph, $l\leq 2$ induced on the vertex set $V(G^1)\backslash V_3$, and $M^2$ is the labeled graph of $H^2$.  Note that the two labeled graphs $M$ and $M^2$, are constructed differently.  $M$ is constructed on the vertex set $V(M)=I^1\backslash N_{G^1}^I(u)$ whereas $M^2$ is the labeled graph on $V(M^2)=I^1$.  We fix $S\subseteq C^1$ to be a minimum Steiner set of $G^1$.  The following theorem characterizes $K_{1,4}$-free $3$-split graphs with $|S|=|I^1|-3$. 
\vspace{-2pt}
\begin{theorem} \label{vi-3}
$|S|=|I^1|-3$ if and only if one of the following is true.  \\
1. $S\cap V_3\neq\emptyset$ and $\alpha(M)=1$. \\
2. $S\cap V_3=\emptyset$ and $\alpha(M^2)=3$. 
\end{theorem}
\begin{proof}
\emph{Necessity:} If $|S|=|I^1|-3$ then we come across the following two cases.\\
\emph{Case 1:} $S\cap V_3\neq\emptyset$.\\
Let $u\in S\cap V_3$  and $S^{'}=S\backslash \{u\}$.  Observe that, in $G^2$, the $l$-split graph, $l\leq 2$ induced on the vertex set $V(G^1)\backslash N_{G^1}^I(u)$, $N_{G^2}^{I}(S^{'})=I^2$. i.e., $|S^{'}|=|S|-1=|I^1|-3-1$ and $|I^2|=|I^1|-3$.  This implies that there exist a matching of size at least $1$ in $M$.  By Lemma \ref{lem2}, $\alpha(M)\leq 2$.   Suppose $\alpha(M)= 2$, then by Theorem \ref{vi-4}, $|S|=|I^1|-4$.  However, we know that $|S|=|I^1|-3$ and therefore $\alpha(M)\neq 2$.  We can therefore conclude that $\alpha(M)\leq 1$.  If $\alpha(M)=0$, then since $G^1$ is connected and $K_{1,4}$-free, $|I^2|=3$.  In this case, $S=\{u\}, |S|=|I^1|-2$.  Therefore, it follows that  $\alpha(M)=1$.\\
\emph{Case 2:} $ S\cap V_3=\emptyset$.\\
Since $S\cap V_3=\emptyset$, $S$ is a minimum Steiner set in $G^1$ and since $H^2$ is the induced on the vertex set $V(G^1)\backslash V_3$, $S$ is also a minimum Steiner set in $H^2$.   Observe that if $\alpha(M^2)=k$, then the size of the minimum Steiner set in $H^2$ is $|V(I_{H}^{2})|-k$ by Lemma \ref{2splitalpha}.  Since $V(I_{H}^{2})=I^1$ we can conclude that $\alpha(M^2)=3$.
\\
\emph{Sufficiency:} 
\emph{Case 1:}  $S\cap V_3\neq\emptyset$ and $\alpha(M)=1$\\
Let $u\in S\cap V_3$ and $\{a,b\}\in E(M)$ be the edge that forms the matching of size $1$, such that label($\{a,b\}$)=$v_{ab}$.  Clearly, $|N_{G^1}^I(X)|=5$, where $X=\{u,v\}$.   Let $Y$ be the corresponding clique set of $I^1\backslash N_{G^1}^I(X)$.
 Observe that $S=X\cup Y$ forms a Steiner set of $G^1$, $|Y|=|I^1|-5$ and  $|S|=|X\cup Y|=|I^1|-3$.  \\
\emph{Case 2:} $S\cap V_3=\emptyset$ and  $\alpha(M^2)=3$\\
Let $\{a,b\}, \{c,d\}, \{e,f\}\in E(M^2)$ be the edges  that form the matching of size $3$, such that label($\{a,b\}$)=$v_{ab}$,  label($\{c,d\}$)=$w_{cd}$,  label($\{e,f\}$)=$x_{ef}$.   Clearly, $|N_{G^1}^I(Y)|=6$, where $Y=\{v,w,x\}$.   Let $Z$ be the corresponding clique set of $I^1\backslash N_{G^1}^I(Y)$.  Observe that $S=Y\cup Z$ forms a Steiner set of $G^1$, $|Z|=|I^1|-6$ and  $|S|=|Y\cup Z|=|I^1|-3$.   This completes the proof.  \hfill\qed 
\end{proof}
\vspace{-5pt}
\begin{theorem} \label{bound}
$|I^1|-4\leq |S|\leq |I^1|-2$.
\end{theorem}
\begin{proof}
The lower bound is true by Theorem \ref{maxn-4}.  Theorem \ref{vi-4}, and Theorem \ref{vi-3} characterizes the $3$-split graphs such that $|S|=|I^1|-4$, and $|S|=|I^1|-3$, respectively.   We shall now look into the upper bound.  Since $G^1$ is $3$-split, there exist $u\in V_3$.  Let $Y$ be the corresponding clique set of $I^1\backslash N_{G^1}^I(u)$.  Observe that $N_{G^1}^I(Y\cup \{u\})=I^1$ and $|S|\leq |Y|+1$.  Since $|Y|\leq|I^1|-3$,  it follows that $|S|\leq |I^1|-2$.  Therefore the theorem. \hfill \qed
\end{proof}
\subsection{Polynomial-time algorithm to find a minimum Steiner tree}  \label{K14algsection} 
Using the structural results presented in Section ${3}$, in this section, we shall present a polynomial-time algorithm to find a minimum Steiner tree in $K_{1,4}$-free split graphs.    Algorithm \ref{alg3} finds a minimum Steiner set $S$ of a given $K_{1,4}$-free split graph $G^0$ with $R\subseteq V(G^0)$ being terminal vertices.  Further, the minimum Steiner tree $T$ is obtained using standard Breadth First Search on $G[R\cup S]$.  \\
We shall now present a sketch of the algorithm and a detailed one is presented in Algorithm \ref{alg3}.  As part of preprocessing, we prune the sets $S_1, S_2,$ and $ S_3$ as defined in Section \ref{section2.1}.  Since $G^1$ is a $K_{1,4}$-free split graph, $G^1$ is $l$-split, $l\leq 3$.   We come across four cases as follows.  If $G^1$ is a $0$-split graph, then $R$ is connected and the minimum Steiner set $S=\emptyset$.  If $G^1$ is a $1$-split graph, then the corresponding clique set of $I^1$ is a minimum Steiner set.  If $G^1$ is a $2$-split graph, then we find the labelled graph $M$ of $G^1$ and the maximum matching $P$ of $M$.  Subsequently, we find the corresponding vertex set $Q\subseteq C$ of the matching $P$ and the corresponding clique set  $Q'$ of $I\backslash N_{G}^I(Q)$.  The minimum Steiner set is $S=Q\cup Q'$ ( from Lemma \ref{2splitalpha} ).   Given a $3$-split graph, we perform a transformation to obtain a $2$-split graph.  We identify the size of a minimum Steiner set and the Steiner set with the help of the labelled graph associated with the transformed $2$-split graph.  Interestingly, based on the matching size, we get to identify the  size of minimum Steiner set and the corresponding clique set helps us to identify the Steiner set.  It is important to highlight the fact that if matching size is $1$, we look at two different labelled graphs to identify the minimum Steiner set.  The detailed  algorithm is presented in Algorithm \ref{alg3}.

\begin{algorithm}[!h]
\caption{Compute\_Steiner\_Tree\_$K_{1,4}$-free($G^0,R$)} \label{alg3}
\begin{algorithmic}[1]
\STATEx {/*$G^0$ \texttt{is the $K_{1,4}$-free split graph and} $R$ $\subseteq$ $V(G^0)$ \texttt{is the set of terminal vertices} */}
\STATE {$G^{1}=$ Pruning($G^0,R$) i.e., $G^{1}=G^0\backslash (S_1\cup S_2\cup S_3)$} \label{prun}
\STATE {Initialize the output Steiner set $S= \emptyset$}
\IF {$G^1$ is a $1$-split graph}
 \STATE{Find corresponding clique set $S$ of $I^1$}  \label{step1split}
\ELSIF {$G^1$ is a $2$-split graph}
\STATE{$S$=Compute\_Steiner\_$2$-split\_graph($G^1$)}
\ELSE
\STATE {$S$=Compute\_Steiner\_$3$-split\_graph($G^1$) }
\ENDIF
\STATE {Obtain the Breadth First Search tree $T$ in the graph induced on vertices $S\cup R$.} \label{st}
\STATE {Output $T$}
\end{algorithmic}
\end{algorithm}
\vspace{-13pt}
\begin{algorithm}[!h]
\caption{Compute\_Steiner\_$2$-split\_graph($G$) } \label{alg2split}
\begin{algorithmic}[1]
\STATEx  { /*$G$ \texttt{is $K_{1,4}$-free and $2$-split graph */ } }
\STATE {Initialize the Steiner set $S^2=\emptyset$}
\STATE {Construct the labeled graph $M$ of $G$}  
 \STATE {Find a maximum matching $P$ in $M$ and find the corresponding vertex set $S^1$ of $P$} \label{step2s_m}
 \STATE {Find the corresponding clique set $S^2$ of the vertex set $I\backslash N_{G}^{I}(S^1)$ }\label{step2s_s2}
 \STATE  {Return Steiner set $S^1\cup S^2$}
\end{algorithmic}
\end{algorithm}

\begin{algorithm}[!h]
\caption{Compute\_Steiner\_$3$-split\_graph($G$) } \label{alg3split}
\begin{algorithmic}[1]
\STATEx  {/*$G=I+C$ \texttt{is $K_{1,4}$-free $3$-split graph}*/}
\STATE {Initialize Steiner set $S^2=\emptyset$, $S^1=\{u\}$ where $u\in V_3$ and edge set $P^1=\emptyset$ }
\FOR {every vertex $v\in V_3$}
\STATE {Find the $2$-split graph $G^2$ induced on $V(G)\backslash N_{G}^I(v)$ and the labeled graph $M$ of $G^2$}
\STATE {Find a maximum matching $P^*$ of $M$} \label{step3s_mat}
\IF {$|P^1|\leq|P^*|$}
\STATE {Update $P^1=P^*$}
\STATE {Update $S^1=\{v\}~\cup$ corresponding vertex set of $P^1$ in $G^2$} \label{step3s_s1}
\ENDIF
\ENDFOR
\IF{$|P^1|<1$}
\STATE {Find the $2$-split graph $H^2$ induced on $V(G)\backslash V_3$} 
\STATE {Find the labeled graph $M^2$ of $H^2$ and a maximum matching $P^2$ of $M^2$}
\IF {$|P^2|=3$}
\STATE{$S^1=$  corresponding vertex set of $P^2$ in $H^2$}  \label{step3s_s1_2}
\ENDIF
\ENDIF
\STATE{Find the corresponding clique set $S^2$ of the vertex set $I\backslash N_{G}^{I}(S^1)$} \label{step3s_s2}
\STATE {Return Steiner set $S^{1}\cup S^2$ }
\end{algorithmic}
\end{algorithm} 
\vspace{-20pt}
\subsubsection{Proof of correctness of Algorithm \ref{alg3}} 
Step \ref{prun} of Algorithm \ref{alg3}  prunes the input graph $G^0$ to obtain $G^1$ and by Lemma \ref{lem1}, an optimal Steiner set of $G^1$ is also an optimal Steiner set of $G^0$.   If $G^1$ is a $1$-split graph, then $|S|=|I^1|$ and our algorithm correctly computes such a Steiner set $S$ in step \ref{step1split}.  If $G^1$ is a $2$-split or $3$-split graph, then Algorithm \ref{alg3} calls Algorithm \ref{alg2split}, or Algorithm \ref{alg3split}, respectively.  Now we shall look into Algorithm \ref{alg2split} in detail.  The algorithm finds the labeled graph $M$ in Step $2$.  
%
%
Note that Algorithm \ref{alg2split} in Step \ref{step2s_m} finds a maximum matching $P$ in $M$, and finds the corresponding vertex set $S^1$ of $P$ such that $|S^1|=|P|$.   Step \ref{step2s_s2} finds $S^2$ such that  $|S^2|=|I^1|-2|P|$.   The Steiner set $S^1\cup S^2$ is returned in Step $5$ where $|S^1\cup S^2|=|I^1|-|P|$, which is correct due to Lemma \ref{2splitalpha}  and hence Algorithm \ref{alg2split} returns an optimum Steiner set. 

In Algorithm \ref{alg3split}, for every $v\in V_3$, we find $G^2$ and its labeled graph $M$ in Step $3$.  A maximum matching on $M$ is obtained in Step $4$.  Step $6$ and \ref{step3s_s1} updates maximum matching $P^1$ and its corresponding vertex set $S^1$ found so far.  Note that by Theorem \ref{bound}, Steiner set $S$ of $G$ is bounded as $|I|-4\leq |S|\leq |I|-2$.  We can see the following cases. \\  \emph{Case (i)} $|S|= |I|-4$.   By Theorem \ref{vi-4}, $|P^1|=2$ and it follows that $|S^1|=3$ and $|N_{G}^I(S^1)|=7$.  Step $17$ finds $S^2$ such that $|S^2|=|I|-7$.  Step $18$ returns $S^1\cup S^2$ where $|S^1\cup S^2|=|I|-4$. \\  \emph{Case (ii)} $|S|= |I|-3$.  By Theorem \ref{vi-3}, either $|P^1|=1$ or $|P^2|=3$.  If $|P^1|=1$, then $|S^1|=2$ and $|N_{G}^I(S^1)|=5$.  Step $17$ finds $S^2$ such that $|S^2|=|I|-5$.  Note that $|S^1\cup S^2|=|I|-3$.  If $|P^2|=3$, then $|S^1|=3$ and $|N_{G}^I(S^1)|=6$.  Step $17$ finds $S^2$ such that $|S^2|=|I|-6$.  Observe $|S^1\cup S^2|=|I|-3$.  \\  \emph{Case (iii)} $|S|= |I|-2$.  It follows that $|P^1|=0$.  Since we initialized $S^1$ with a vertex $u\in V_3$, $|S^1|=1$ and $|N_{G}^I(S^1)|=3$.  Observe that $|S^2|=|I|-3$ and $|S^1\cup S^2|=|I|-2$.   This completes the case analysis and Algorithm \ref{alg3split} correctly computes a Steiner set of $G$.  Therefore, Algorithm \ref{alg3} correctly computes the minimum Steiner tree in Step $10$.
\vspace{-20pt}
\subsubsection{Run-time analysis of Algorithm \ref{alg3}} \noindent \\
Let $n, m$ represents the size of vertex set, and the edge set, respectively of the input graph $G^0$.  We shall first analyze the run-time of Algorithm \ref{alg2split} and Algorithm \ref{alg3split} as Algorithm \ref{alg3} invokes Algorithm \ref{alg2split} or Algorithm \ref{alg3split} at Steps $6$, $8$, respectively.   For Algorithm \ref{alg2split}, observe that creation of the labeled graph in Step $2$ needs $O(n)$ effort as $|E(M)|+|V(M)|=O(n)$.  Step $3$ finds a maximum matching of $M$ which can be done in $O(n^{\frac{3}{2}})$ time using general graph maximum matching algorithm \cite{maxmatching}.  Corresponding clique set in Step $4$ can be found in $O(n)$ time.  Therefore, the run-time of Algorithm \ref{alg2split} is $O(n^{\frac{3}{2}})$.   Consider Algorithm \ref{alg3split}, Steps $3$ to  $8$ are iterated at most $n$ times.  Step $3$ needs $O(n)$ effort.  Finding a matching of $M^2$ in step \ref{step3s_mat} needs $O(n^{\frac{3}{2}})$ time.  Steps $6,7$ incurs constant effort.  Therefore, the iteration of Steps $3$ to $8$ involves  $O(n^{\frac{5}{2}})$ effort.  Note that Steps $11$, $12$ need $O(n)$, $O(n^{\frac{3}{2}})$, respectively and Step \ref{step3s_s1_2} incurs a $O(n)$ effort.  Finding $S^2$ in step \ref{step3s_s2} can be done in $O(n)$ time.  Overall, the run-time of Algorithm \ref{alg3split} is $O(n^{\frac{5}{2}})$.

Now we shall discuss run time of Algorithm \ref{alg3}.  Pruning of verices in step \ref{prun} of Algorithm \ref{alg3} takes $O(n.\Delta)$ effort where $\Delta$ denotes maximum degree of the input graph $G$.  Step \ref{step1split} takes $O(n)$ time.  Steps $6$, $8$ takes $O(n^{\frac{3}{2}})$ time, $O(n^{\frac{5}{2}})$ time, respectively.  Step $10$ incurs $O(n+m)$ time.  Therefore the run time of Algorithm \ref{alg3} is $O(n^{\frac{5}{2}})$.  Thus, Steiner tree in $K_{1,4}$-free split graph is polynomial-time solvable.
%
\vspace{-10pt}
\section{Steiner tree in $K_{1,5}$-free Split Graphs is NP-complete}
In the earlier section, we have presented a polynomial-time algorithm for Steiner tree in $K_{1,4}$-free split graphs.  In this section, we present the other half of the dichotomy, which is to show that Steiner tree in $K_{1,5}$-free split graph is NP-complete.  Interestingly, the reduction presented in \cite{white} generates instances of $K_{1,5}$-free split graphs.  For the sake of completeness, we present our observations along with proofs.  Towards this attempt, we recall the classical problem Exact 3 cover \cite{karp} which is a candidate NP-complete problem for our investigation.

\begin{center}
\fbox{\parbox[c][][c]{0.9\textwidth}{    
\emph{Exact-3-cover(Z,T)}\\
Instance:	A Collection $T$ of $3$ element subsets of a set $Z=\{u_1,u_2,\ldots,u_{3q}\}$.\\
Question:	Is there a sub collection $T^{'}$ $\subseteq$ $T=\{c_1,c_2,\ldots,c_n\}$ such that for every $u_i\in Z$, $1\leq i\leq 3q$ $u_i$ belongs to exactly one member of $T^{'}$?} }	
\end{center}
%
%
%
We recall the decision version of Steiner tree problem, restricted to $K_{1,5}$-free split graphs.

\begin{center}
\fbox{\parbox[c][][c]{0.8\textwidth}{    
\emph{Steiner tree(G,R,k)}\\
Instance:	$K_{1,5}$-free Split Graph $G(V,E)$, Terminal Set $R\subseteq V(G)$, Integer $k\geq0$\\
Question:	Is there a set $S\subseteq V(G)\backslash R$ such that $|S|\leq k$ and $G[S\cup R]$ is connected?} }	
\end{center}

\begin{theorem} \label{k15stree}
Steiner tree problem in $K_{1,5}$-free split graph is NP-complete.
\end{theorem}
\begin{proof} \textbf{Steiner tree is in NP}  Given a certificate $S=(G,R,k)$, we show that there exist a deterministic polynomial-time algorithm for verifying the validity of the certificate $S$.   Note that the standard Breadth First Search algorithm can be employed to check whether $S\cup R$ is connected.  $|S|=k$ can be verified in linear time and therefore, overall certificate verification need $O(n+m)$ time, where $n=|V(G)|,~m=|E(G)|$.  	Therefore, we can conclude that Steiner tree is in NP. \\
\textbf{Steiner tree is NP-Hard}
An instance of Exact $3$ cover(Z,T) is reduced to an instance of Steiner tree (G,R,k)  problem as follows:  $I=Z$, $C=\{v_i~|~c_i\in T\}$, $1\leq i\leq n$ and $V(G)=I\cup C$.  Informally, for every element $u\in Z$, create a vertex $u$ such that $u\in I$.  For every member $c_i\in T$, create a vertex $v_i$ such that $v_i\in C$.   $E(G)= \{\{v_i,v_j\}~|~v_i,v_j\in C\}$, $1\leq i\neq j\leq n$ $\cup$ $\{\{v_l,u\}~|~v_l\in C, u\in I$, and $u\in c_l\}$.    $R=I$ and $k=\frac{|Z|}{3}$.   In this reduction, $|V(G)|=|Z|+|T|$ and $|E(G)|=\binom{|T|}{2}+3|T|$.   The above construction is therefore polynomial to the size of input.   We now show that instances created by this reduction are $K_{1,5}$-free split graphs.  On the contrary, assume that there exist a $K_{1,5}$ induced on vertices $\{u,v,w,x,y,z\}$.  Note that at most two vertices (say $u,v$) from clique $C$ can be included in the $K_{1,5}$.  Clearly, $w,x,y,z\in I$ and without loss of generality, $d_{G}^{I}(v)=4$.  This implies that there exist a $4$ element subset $c\in T$ corresponding to the clique vertex $v\in C$, which is a contradiction as all subsets are of size $3$ in collection $T$.  Therefore it follows that the reduced graph $G$ is $K_{1,5}$-free split graph. We now show that there exist an Exact-3-cover(Z,T) if and only if there exist a Steiner tree(G,R,k) in the reduced graph $G$ on at most $k$ Steiner vertices.  For \emph{Necessity:} If there exist $T^{'}\subseteq T$, $|T^{'}|=\frac{|Z|}{3}$ which covers all the elements of $Z$, then the set of vertices $S=\{v\in C~|~c\in T^{'}\}$ where $v$ is the corresponding vertex of $c$ forms a Steiner set in $G$ as $R=Z$.  Also note that $|S|=\frac{|Z|}{3}$.  For \emph{Sufficiency:} If there exist a Steiner set $S\subseteq C$ in the reduced graph $G$ on at most $k=\frac{|Z|}{3}$ Steiner vertices, then observe that for all vertex $v\in S$, $d_{G}^I(v)=3$, $|S|= \frac{|Z|}{3}$ and $|N_{G}^I(S)|=|Z|$.  It follows that there does not exist $u,v\in S$ such that $N_{G}^I(u)\cap N_{G}^I(v)\neq\emptyset$.   Therefore, $T^{'}=\{c\in T~|~v\in S\}$ where $v$ is the corresponding vertex of $c$ forms an exact 3 cover of $Z$.  This completes the proof of the claim. 
We can conclude that Steiner tree problem is  NP-complete in $K_{1,5}$-free split graphs.  \hfill \qed
\end{proof}

\section {Conclusions and Future Work}
We have presented an interesting dichotomy result that Steiner tree problem is polynomial-time solvable in $K_{1,4}$-free split graphs and NP-complete in $K_{1,5}$-free split graphs.  This result is tight and it identifies the right gap between NP-completeness and polynomial-time solvability of Steiner tree in split graphs.  Using the structural results presented here, an interesting direction for further research would be to explore the complexity of other classical problems which are NP-complete restricted to split graphs.

\bibliographystyle{splncs1}
\bibliography{steinerref}


\end{document}